\documentclass[9pt,twocolumn,twoside]{IEEEtran}
\usepackage{textcomp}
\def\BibTeX{{\rm B\kern-.05em{\sc i\kern-.025em b}\kern-.08em
    T\kern-.1667em\lower.7ex\hbox{E}\kern-.125emX}}

\usepackage{amsmath} %
\usepackage{amssymb}  %

\usepackage{amsthm}
\usepackage{algorithm}
\usepackage[noend]{algpseudocode}
\usepackage{caption}
\usepackage{subcaption}
\let\labelindent\relax
\usepackage{enumitem}

\usepackage{tikz}
\usepackage{pgfplots}
\usetikzlibrary{external}
\DeclareMathOperator{\tr}{tr}

\newtheorem{theorem}{Theorem}

\newtheorem{lemma}[theorem]{Lemma}
\newtheorem{corollary}[theorem]{Corollary}
\theoremstyle{definition}

\newtheorem{assumption}[theorem]{Assumption}
\newtheorem{remark}[theorem]{Remark}

\usetikzlibrary{intersections,fillbetween,backgrounds}

\title{Safe and Efficient Switching Mechanism Design for Uncertified Linear Controller}

\author{Yiwen Lu and Yilin Mo%
\thanks{This work is supported by the National Key Research and Development Program of China under Grant 2018AAA0101601. The authors are with the Department of Automation and BNRist, Tsinghua University, Beijing, P.R.China. Emails: {\tt\small luyw20@mails.tsinghua.edu.cn, ylmo@tsinghua.edu.cn}.}
}

\begin{document}

\maketitle
\thispagestyle{empty}
\pagestyle{empty}

\begin{abstract}
	Sustained research efforts have been devoted to learning optimal controllers for linear stochastic dynamical systems with unknown parameters, but due to the corruption of noise, learned controllers are usually uncertified in the sense that they may destabilize the system. To address this potential instability, we propose a ``plug-and-play'' modification to the uncertified controller which falls back to a known stabilizing controller when the norm of the difference between the uncertified and the fall-back control input exceeds a certain threshold. We show that the switching strategy is both safe and efficient, in the sense that: 1) the linear-quadratic cost of the system is always bounded even if original uncertified controller is destabilizing; 2) in case the uncertified controller is stabilizing, the performance loss caused by switching converges super-exponentially to $0$ for Gaussian noise, while the converging polynomially for general heavy-tailed noise. Finally, we demonstrate the effectiveness of the proposed switching strategy via numerical simulation on the Tennessee Eastman Process.
\end{abstract}

\section{Introduction}

Learning a controller from noisy data for an unknown system  has been a central topic to adaptive control and reinforcement learning~\cite{aastrom2013adaptive,bertsekas2019reinforcement,recht2019tour,de2021low} for the past decades.
A main challenge to directly applying the learned controllers to the system is that they are usually \emph{uncertified}, in the sense it can be very difficult to guarantee the stability of such controllers due to process and measurement noise. One way to address this challenge is to deploy an additional safeguard mechanism. In particular, assuming the existence of a known stabilizing controller, empirically the safeguard may be implemented by falling back to the stabilizing controller from the uncertified controller, when potential safety breach is detected.

Motivated by the above intuition, this paper proposes such a switching strategy, provides a formal safety guarantee and quantifies the performance loss incurred by the safeguard mechanism, for discrete-time Linear-Quadratic Regulation (LQR) setting with independent and identically distributed process noise with bounded fourth-order moment.
We assume the existence of a known stabilizing linear feedback control law $u = K_0x$, which can be achieved either when the system is known to be open-loop stable (in which case $K_0 = 0$), or through adaptive stabilization methods~\cite{byrnes1984adaptive,faradonbeh2018finite}.
Given an uncertified linear feedback control gain $K_1$,
a modification to the control law $u = K_1 x$ is proposed:
the controller normally applies $u = K_1 x$, but falls back to $u = K_0 x$ for $t$ consecutive steps once $\| (K_1 - K_0) x \|$ exceeds a threshold $M$. The proposed strategy is analyzed from both stability and optimality aspects.
In particular, the main results include:
\begin{enumerate}
    \item We prove the LQ cost of the proposed controller is always \emph{bounded}, even if $K_1$ is destabilizing. This fact implies that the proposed strategy enhances the safety of the uncertified controller by preventing the system from being catastrophically destabilized.
    \item Provided $K_1$ is stabilizing, and $M,t$ are chosen properly, we compare the LQ cost of the proposed strategy with that of the linear feedback control law $u = K_1 x$, and quantify the maximum increase in LQ cost caused by switching w.r.t. the strategy hyper-parameters $M, t$ as merely $O(t^{1/4}\exp(-\mathrm{constant}\cdot M^2))$ in the case of Gaussian process noise, which decays \emph{super-exponentially} as the switching threshold $M$ tends to infinity. We also discuss the extension to general noise distributions with bounded fourth-order moments, where the above asymptotic performance gap becomes $O(t^{1/4}M^{-1})$. 
\end{enumerate}
The performance of the proposed switching scheme is further validated by simulation on the Tennessee Eastman Process example. We envision that the switching framework could be potentially applicable in a wider range of learning-based control settings, since it may combine the good empirical performance of learned policies and the stability guarantees of classical controllers, and the ``plug-and-play'' nature of the switching logic may minimize the required modifications to existing learning schemes.

A preliminary version of this paper~\cite{arxiv_version} has been submitted to IEEE CDC 2022. The main contributions of the current manuscript over the conference submission are: i) the switching scheme has been redesigned, such that the upper bound on LQ cost (Theorem~\ref{thm:bounded_cost}) no longer depends on $K_1$; ii) the conclusions have been extended to noise distributions with bounded fourth-order moments; iii) proofs of all theoretical results are included in the current version of the manuscript.

\subsection*{Related Works}

\paragraph*{Switched control systems}
Supervisory algorithms have been developed to stabilize switched linear systems~\cite{cheng2005stabilization,sun2005analysis,zhang2010asynchronously}, and other nonlinear systems that are difficult to stabilize globally with a single controller~\cite{prieur2001uniting,el2005output,battistelli2012supervisory}. However, most of the paper focuses on the stability of the switched system, while the (near-)optimality of the controllers are less discussed.
Building upon this vein of literature, the idea of switching between certified and uncertified controllers to improve performance was proposed in~\cite{wintz22global}, whose scheme guarantees global stability for general nonlinear systems under mild assumptions. However, no quantitative analysis of the performance under switching is provided. In contrast, we specialize our results for linear systems and prove that switching may induce only negligible performance loss while ensuring safety. 

\paragraph*{Adaptive LQR} Adaptive and learned LQR has drawn significant research attention in recent years, for which high-probability estimation error and regret bounds have been proved for methods including optimism-in-face-of-uncertainty~\cite{abbasi2011regret,cohen2019learning}, thompson sampling~\cite{abeille2018improved}, policy gradient~\cite{fazel2018global}, robust control based on coarse identification~\cite{dean2018regret} and certainty equivalence~\cite{mania2019certainty,faradonbeh2020optimism,faradonbeh2020adaptive,simchowitz2020naive}.
All the above approaches, however, involve applying a linear controller learned from finite noise-corrupted data, which has a nonzero probability of being destabilizing.
Furthermore, given a fixed length of data, the failure probabilities of the aforementioned methods depend on either unknown system parameters or statistics of online data, which implies the failure probability cannot be determined \textit{a priori}, and hence it can be challenging to design an algorithm that strictly satisfies a pre-defined specification of safety.
In~\cite{wang2021exact}, a ``cutoff'' method similar to the switching strategy described in the present paper is applied in an attempt to establish almost sure guarantees for adaptive LQR, which are nevertheless asymptotic in nature, and the extra cost caused by switching is not analyzed.
By contrast, this manuscript provides both non-asymptotic and asymptotic bounds for the switching strategy. 

\paragraph*{Nonlinear controller for LQR} Nonlinearity in the control of linear systems has been studied mainly due to practical concerns such as saturating actuators. The performance of LQR under saturation nonlinearity has been studied in~\cite{gokcek2000slqr,gokcek2001lqr,ossareh2016lqr}, which are all based on stochastic linearization, a heuristics that replaces nonlinearity with approximately equivalent gain and bias. By contrast, the present paper treats nonlinearity as a design choice rather than a physical constraint, and provides rigorous performance bounds without resorting to any heuristics. 

\subsection*{Outline}

The remainder of this paper is organized as follows: Section~\ref{sec:problem} introduces the problem setting and describes the proposed switching strategy. The main results are provided in Section~\ref{sec:theory} and Section~\ref{sec:extension} for Gaussian process noise and noise with bounded fourth-order moments respectively. Section~\ref{sec:simulation} validates the performance of the proposed strategy with a industrial process example. Finally, Section~\ref{sec:conclusion} concludes the paper.

\subsection*{Notations}

The set of nonnegative integers are denoted by $\mathbb{N}$, and the set of positive integers are denoted by $\mathbb{N}^*$. For a square matrix $M$, $\rho(M)$ denotes the spectral radius of $M$, and $\tr(M)$ denotes the trace of $M$. For a real symmetric matrix $M$, $M\succ 0$ denotes that $M$ is positive definite.
$\|v\|$ denotes the 2-norm of a vector $v$ and $\|M\|$ is the induced 2-norm of the matrix $M$, i.e., its largest singular value. For $P\succ 0$, $\langle v, w \rangle_P = v^T P w$ is the $P$-inner product of vectors $v,w$, and $\| v \|_P=\|P^{1/2}v\|$ is the $P$-norm of a vector $v$. For two positive semidefinite matrices $P\succ 0, Q \succ 0$, $\|Q\|_P = \|P^{-1/2}QP^{-1/2}\|= \sup_{\|v\|_P=1}\|v\|_Q^2$. For a random vector $X$, $X\sim \mathcal{N}(\mu,\Sigma)$ denotes $X$ is Gaussian distributed with mean $\mu$ and covariance $\Sigma$. $\mathbb{P}(\cdot)$ denotes the probability operator, $\mathbb{E}(\cdot)$ denotes the expectation operator, and $\mathbf{1}_{E}$ is the indicator function of the random event $E$. For functions $f(x), g(x)$ with non-negative values, $f(x) = O(g(x))$ means $\limsup_{x\to\infty} f(x) / g(x) < \infty$, and $f(x) = \Theta(g(x))$ means $f(x) = O(g(x))$ and $g(x)=O(f(x))$.

\section{Problem Formulation and Proposed Switching Strategy}\label{sec:problem}

Consider the following discrete-time linear plant:
\begin{equation}
    x_{k+1} = Ax_k + Bu_k + w_k,
    \label{eq:dynamics}
\end{equation}
where $k\in \mathbb{N}$ is the time index, $x_k \in \mathbb{R}^n$ is the state vector, $u_k \in \mathbb{R}^m$ is the input vector, and $w_k \in \mathbb{R}^n$ is the process noise. Without loss of generality, the system is assumed to be controllable. We further assume that the initial state $x_0 = 0$, and that $\{ w_k \}$ are independent and identically distributed with covariance matrix $W\succ 0$. %

We measure the performance of a controller $u_k(x_{0:k})$ in terms of the infinite-horizon quadratic cost defined as:
\begin{equation}
    J=\limsup _{T \rightarrow \infty} \frac{1}{T} \mathbb{E}\left[\sum_{k=0}^{T-1} x_{k}^{T} Q x_{k}+u_{k}^{T} R u_{k}\right],
    \label{eq:J}
\end{equation}
where $Q \succ 0, R \succ 0$ are fixed weight matrices specified by the system operator.
It is well known that the optimal controller is the linear feedback controller of the form $u(x) = K^*x$, where the optimal gain $K^*$ can be determined by solving the discrete-time algebraic Riccati equation.

In this paper, we assume that the system and input matrices $A,B$ are unavailable to the system operator, and hence she cannot determine the optimal feedback gain $K^*$. Instead, she has the following two feedback gains:
\begin{itemize}
    \item \emph{Primary gain $K_1$}, typically learned from data, which can be close to $K^*$ but does not have stability guarantees;
    \item \emph{Fallback gain $K_0$}, which is typically conservative but always guaranteed to be stabilizing, i.e., $\rho(A+BK_0) < 1$.
\end{itemize}

Ideally, the system operator would want to use $K_1$ as much as possible, as it usually admits a better performance. However, since $K_1$ is not necessarily stabilizing, a switching strategy is deployed in pursuit of both safety and performance of the system. The block diagram of the closed-loop system under the proposed switching strategy is shown in Fig.~\ref{fig:block_diagram}, and the switching logic is described in Algorithm~\ref{alg:main}. In plain words, the proposed switched control strategy is normally applying $u = K_1 x$, while falling back to $u = K_0 x$ for $t$ consecutive steps once $\| (K_1 - K_0) x \|$ exceeds a threshold $M$.

\begin{figure}[!htbp]
    \centering
    \includegraphics{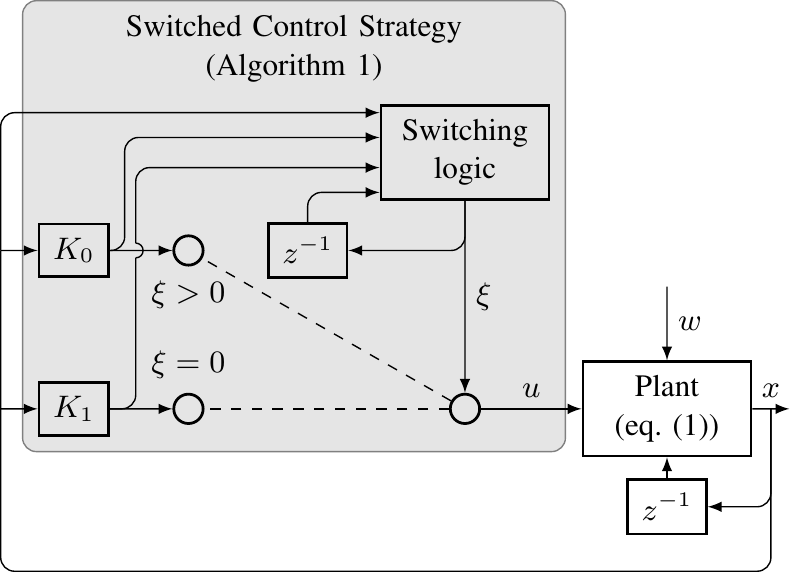}
    \caption{Block diagram of the closed-loop system under the proposed switching strategy. The controller selects $u = K_1 x$ when $\xi = 0$ and $u = K_0 x$ when $\xi > 0$, where $\xi$ is an internal counter determined by the switching logic.}
    \label{fig:block_diagram}
\end{figure}

\renewcommand{\algorithmicrequire}{\textbf{Input:}}
\renewcommand{\algorithmicensure}{\textbf{Output:}}
\begin{algorithm}[!htbp]
    \begin{algorithmic}[1]
        \Require Current state $x$, primary gain $K_1$, fallback gain $K_0$, current counter value $\xi$, switching threshold $M$, dwell time $t$ 
        \Ensure Control input $u$, next counter value $\xi'$
        \If{$\xi > 0$}
            \State $u \gets K_0 x$
        \Else
            \If{$ \| (K_1 - K_0) x \|  \geq M$}
                \State $\xi \gets t, u \gets K_0 x$
            \Else
                \State $u \gets K_1 x$
            \EndIf
        \EndIf
        \State $\xi' \gets \max\{\xi - 1, 0 \}$
    \end{algorithmic}
    \caption{Proposed switched control strategy}
    \label{alg:main}
\end{algorithm}

\section{Main Theoretical Results}\label{sec:theory}

This section is devoted to proving the stability of the proposed switching strategy as well as quantifying performance loss it incurs.
It is assumed throughout this section that the process noise obeys a Gaussian distribution, i.e., $w_k \sim \mathcal{N}(0, W)$.

\subsection{Upper Bound on the LQ Cost}

In this subsection, we prove that when $M$ and $t$ are fixed, the LQ cost associated with the proposed switched controller is always bounded, regardless of the choice of the underlying primary gain $K_1$ and hyper-parameters $M, t$. Notice that naively implementing the linear controller $u = K_1 x$ without the switching results in an infinite LQ cost when $K_1$ is destabilizing. As a result, the proposed switching strategy ensures the stability of the closed-loop system.

Since the fallback gain $K_0$ is stabilizing, there exists $P_0 \succ 0$ that satisfies the discrete-time Lyapunov equation
\begin{equation}
    (A+BK_0)^T P_0(A+BK_0) - P_0 + Q + K_0^T R K_0 = 0,
    \label{eq:P0}
\end{equation}
and hence there exists $0<\rho_0 < 1$ such that
\begin{equation}
    (A+BK_0)^T P_0(A+BK_0) \prec \rho_0P_0.
    \label{eq:rho0}
\end{equation}
The following lemma constructs a quadratic Lyapunov function from $P_0$ and states that the Lyapunov function has bounded expectation:
\begin{lemma}
    Assuming that $\rho_0, P_0$ satisfy~\eqref{eq:P0},
    let $V_{0,k} = x_k^T P_0 x_k$, then it holds for any $k$ that
    \begin{equation}
        \mathbb{E} V_{0,k}\leq \frac{4(1+\rho_0)(M^{2} \| B \|^{2} \| P_0 \| + \tr(WP_0))}{(1-\rho_0)^2}.
        \label{eq:EV}
    \end{equation}
    \label{lemma:EV}
\end{lemma}

Based on the fact that the LQ cost can be upper bounded in terms of the expectation of the quadratic Lyapunov function $V_{0,k}$ defined above, we have the following theorem:

\begin{theorem}
    Assuming that $P_0, \rho_0$ satisfy~\eqref{eq:P0} and~\eqref{eq:rho0}, the LQ cost defined in~\eqref{eq:J} satisfies
    \begin{equation}
        J \leq \left(\frac{8(1+\rho_0)\| B \|^2 \| P_0 \|}{(1 - \rho_0)^2} + 2\|R\|\right)M^2+\frac{8(1+\rho_0)\tr(WP_0)}{(1 - \rho_0)^2}.
        \label{eq:J_unstable}
    \end{equation}
    \label{thm:bounded_cost}
\end{theorem}

\begin{proof}
    By definition of $J$, we only need to prove $\mathbb{E}(\| x_k \|_Q^2 + \| u_k \|_R^2)$ is not greater than the RHS of~\eqref{eq:J_unstable} for any $k$. Notice that
    \begin{align}
        & \| x_k \|_Q^2 + \| u_k \|_R^2 = \| x_k \|_{Q + K_0^T R K_0}^2 + \| u_k \|_R^2 - \| K_0x_k \|_{R}^2 \nonumber \\
        & \leq \| x_k \|_{P_0}^2 + \| u_k \|_R^2 - \| K_0x_k \|_{R}^2
        \label{eq:stage_cost}
    \end{align}
    By the switching strategy, it holds $u_k = K_0 x_k + \Delta u_k$, where $\Delta u_k:=(K_1 - K_0) x_k \mathbf{1}_{\{u_k = K_1 x_k\}}$ satisfies $\| \Delta u_k \| \leq M$, and hence,
    \begin{equation*}
        \| u_k \|_R \leq \| K_0 x_k \|_R + M\|R \|^{1/2},
    \end{equation*}
    which implies
    \begin{align}
        & \| u_k \|_R^2 - \| K_0x_k \|_{R}^2 \leq 2M\|R \|^{1/2}\| K_0 x \|_R + M^2 \| R \| \nonumber \\
        &\leq 2M \| R \|^{1/2}  \| x_k \|_{P_0} + M^2 \|R \|.
        \label{eq:u_cost_diff}
    \end{align}
    Substituting~\eqref{eq:u_cost_diff} into~\eqref{eq:stage_cost}, we get
    \begin{equation}
        \| x_k \|_Q^2 + \| u_k \|_R^2 \leq (\| x_k \|_{P_0} + M \| R \|^{1/2})^2 \leq 2(\| x_k \|_{P_0}^2 + M^2 \| R \|).
        \label{eq:stage_cost_2}
    \end{equation}
    Taking the expectation on both sides of~\eqref{eq:stage_cost_2} and applying Lemma~\ref{lemma:EV}, the conclusion follows.
\end{proof}

\subsection{Upper bound on performance loss caused by switching}
\label{sec:gap}

In this subsection, we quantify the extra LQ cost caused by the conservativeness of switching when $K_1$ is stabilizing. Let $J^{K_1}$ denote the LQ cost associated with the closed-loop system under the linear controller $u = K_1 x$, and $J^{K_1,M,t}$ denote the LQ cost associated with the closed-loop system under our proposed switched controller with primary gain $K_1$ and hyper-parameters $M,t$.
To quantify the behavior of the system under switching, we resort to a common quadratic Lyapunov function for $A+BK_1$ and $(A+BK_0)^t$, which always exists for sufficiently large dwell time $t$. Formally speaking, the following inequalities holds:
\begin{equation}
    \begin{cases}
        (A+BK_1)^T P (A+BK_1) \prec \rho P,\\
        ((A+BK_0)^t)^T P (A+BK_0)^t \prec \rho P,
    \end{cases}
    \label{eq:common_lyap}
\end{equation}
where $0 < \rho < 1$ and $P \succ 0$.
Notice that $\rho, P$ that satisfy the first inequality always exist due to the stability of $A+BK_1$, and given $\rho, P$, the quantity $t$ that satisfy the second inequality exists since $\lim_{t\to\infty}(A+BK_0)^t = 0$ by the stability of $A+BK_0$.

Before proving the main result, we need a supporting theorem, which quantifies the tail bound for an exponentially weighted sum of potentially dependent random variables with Gaussian-like tails:  

\begin{theorem}
    Let $\{ X_i \}$ be a sequence of  random variables that satisfy
    $
    \mathbb{P}( X_i \geq a) \leq C_1 \exp(-C_2a^2)
    $
    for any $i = 0,1,\ldots$ and any $a > 0$, where $C_1, C_2$ are positive constants.
	Let $S_k = \sum_{i = 0}^k \varrho^{k-i} X_i$, where $0 < \varrho < 1$, then for any $a\geq 2C_2^{-1/2}(1 - \varrho^{1/2})^{-1}$, it holds
    \begin{equation}
	    \mathbb{P}( S_k  \geq a) \leq \tilde{C}_1 \exp\left( -\tilde{C}_2a^2 \right),\label{eq:tailbound}
    \end{equation}
    where $\tilde{C_1} = 2C_1 ({\min\{\varrho^{-1} - 1, 1\}})^{-1}$, $\tilde{C}_2 = {(1 - \varrho^{1/2})^2}\allowbreak C_2$.
    \label{thm:exp_weight_sum}
\end{theorem}

\begin{remark}
	Notice that if $X_i$s are jointly Gaussian distributed, then \eqref{eq:tailbound} can be trivially proved by computing the covariance of $S_k$, which is also Gaussian. In essence, Theorem~\ref{thm:exp_weight_sum} serves as an extension of the result for jointly Gaussian random variables, by allowing non-Gaussian random variables with Gaussian-like tail distribution, and removing any restriction on the joint distribution between random variables.
\end{remark}

By leveraging Theorem~\ref{thm:exp_weight_sum}, we can bound the fourth moment of the state $x_k$ as well as the probability of the switching:

\begin{theorem}
    Assume that $P_0, \rho_0$ satisfy~\eqref{eq:P0} and~\eqref{eq:rho0}, and that $\rho, P, t$ satisfy~\eqref{eq:common_lyap}. Let $\tilde{W} = \sum_{\tau=0}^\infty (A+BK_0)^\tau W((A+BK_0)^\tau)^T$, $\mathcal{K} = \| K_1 - K_0 \|$, and $a_0 = (8n \| \tilde{W} \| \| P \| \| P^{-1} \|)^{1/2}(1 - \rho^{1/4})^{-1}$. If the threshold $M \geq a_0\mathcal{K}$ is large enough, then the following statements hold:
    \begin{enumerate}
        \item The fourth moments of $x_k$ is bounded:
        \begin{equation}
            \mathbb{E} \| x_k \|^4_{P_0} \leq 8(\mathcal{Q} \| P_0 \|_P^2+(n^2 + 2n)\| P_0 \|_{\tilde{W}^{-1}}^2),
        \end{equation}
        where 
        \begin{equation*}
            \mathcal{Q} = \frac{6\rho (\tr(\tilde{W}P)^2) + (1 - \rho)(n^2+2n)\| P \|_{\tilde{W}^{-1}}^2}{(1 - \rho)(1 - \rho^2)},
        \end{equation*}
        \item The probability of not using feedback gain $K_1$ satisfies:
        \begin{equation*}
            \mathbb{P}(u_k \neq K_1 x_k) \leq t \mathcal{E}(M / \mathcal{K}),
        \end{equation*}
        where
        \begin{equation*}
            \mathcal{E}(a) = \frac{4n}{\rho^{-1 / 2}-1} \exp \left(-\frac{(1-\rho^{1 / 4})^{2}}{2n \| \tilde{W} \| \| P \| \| P^{-1} \| } a^{2}\right),
        \end{equation*}
        which decays {\bf super-exponentially} w.r.t. the threshold $M$.
    \end{enumerate}
    \label{thm:moment_prob}
\end{theorem}

We are now ready to state the main theorem of this subsection:

\begin{theorem}
    With $\rho_0, P_0, \rho, P, a_0, \tilde{W}, \mathcal{K}, \mathcal{Q}, \mathcal{E}$ defined the same as in Theorem~\ref{thm:moment_prob}, assuming that the dwell time $t$ satisfies \eqref{eq:common_lyap} and the threshold $M \geq a_0 \mathcal{K}$, it holds that
    \begin{equation}
        J^{K_1,M,t} - J^{K_1} \leq 2 C_1C_2\mathcal{G}+(C_2^2+C_3) \mathcal{G}^2,
        \label{eq:gap}
    \end{equation}
    where
    \begin{align*}
        & \mathcal{G} = C_4 (t\mathcal{E}(M / \mathcal{K}))^{1/4},  \quad C_{1}=\sqrt{\tr(WP) \| Q_1 \|_P/ (1 - \rho)}, \\
        & C_2 = \| \Delta_1 \| \| Q_1 \|_{P_0}\sum_{s=0}^\infty \| A_1^s \|_{Q_1}, \quad C_3 = \| \Delta_2 \| \| P_0^{-1} \|, \\
        & C_4 = 2^{3/4}(\mathcal{Q} \| P_0 \|_P^2+(n^2 + 2n)\| P_0 \|_{\tilde{W}^{-1}}^2)^{1/4}, \\
        & {Q}_1 = Q + K_1^TRK_1, A_1 = A+BK_1, \\ & \Delta_1 = B(K_0 - K_1), \quad \Delta_2 = K_0^T R K_0 - K_1^T R K_1.
    \end{align*}
    \label{thm:gap}
\end{theorem}

\begin{proof}
    Let $\check{x}_0 = x_0$ and $\check x_{k+1}= A_1\check x_k + w_k$, then
    $$
    J^{K_1} = \lim_{T \to \infty}\frac1T\sum_{k = 0}^{T-1}\mathbb{E}\| \check{x}_k \|_{Q_1}^2.$$
    On the other hand, we have $$
    J^{K_1,M,t} = \limsup_{T \to \infty}\frac1T\sum_{k = 0}^{T-1}\mathbb{E}[x_k^T Q x_k + u_k^T R u_k],$$ and therefore we only need to prove $\mathbb{E}[x_k^T Q x_k + u_k^T R u_k] - \mathbb{E}\| \check{x}_k \|_{Q_1}^2$ is no greater than the RHS of~\eqref{eq:gap} for any $k$. Notice that
    \begin{equation*}
        x_k^T Q x_k + u_k^T R u_k - \|\check{x}_k \|_{Q_1}^2 = \| x_k \|_{Q_1}^2  - \|\check{x}_k \|_{Q_1}^2 + x_k^T \Delta_2 x_k \mathbf{1}_{F_k},
    \end{equation*}
    where $F_{k} = \{ u_k \neq K_1 x_k \}$ denotes the event that the fallback mode is active and the gain $K_1$ is not applied at step $k$. We will next bound $\mathbb{E}\| x_k \|_{Q_1}^2  - \mathbb{E}\|\check{x}_k \|_{Q_1}^2$ and $\mathbb{E}(x_k^T \Delta_2 x_k \mathbf{1}_{F_k})$ respectively.

    \paragraph{Bounding $\mathbb{E}\| x_k \|_{Q_1}^2  - \mathbb{E}\|\check{x}_k \|_{Q_1}^2$} \label{par:bound_norm_diff} Notice that 
    \begin{equation*}
    	x_k = A_1x_{k-1}+w_{k-1}+ \Delta_1 x_{k-1}\mathbf{1}_{F_{k-1}},
    \end{equation*}
    and by recursively applying this expansion, we get
    \begin{align*}
	    x_k & = A_1^{k} x_{0}+\sum_{s=0}^{k-1}A_1^{k-s-1} (w_{s} + \Delta_1 {x}_{s} \mathbf{1}_{F_{s}})\\
        & = \check x_k +\sum_{s=0}^{k-1}A_1^{k-s-1}  \Delta_1 {x}_{s} \mathbf{1}_{F_{s}}.
    \end{align*}
    Hence,
    \begin{equation*}
        \| x_k \|_{Q_1} \leq  \| \check{x}_k \|_{Q_1} + \| \Delta_1 \|_{Q_1} \sum_{s=0}^{k-1}\|A_1^{k-s-1}\|_{Q_1} \| {x}_{s} \|_{Q_1}\mathbf{1}_{F_{s}}.
    \end{equation*}
    From the fact that $\mathbb{E}(\sum_{i=1}^n X_i)^{2} \leq (\sum_{i=1}^n \sqrt{\mathbb{E}X_i^2})^2$ for any random variables $X_1,\ldots,X_n$, we have
    \begin{align*}
        & \mathbb{E}\| x_k \|_{Q_1}^2 \leq \left(\vphantom{\sum_1^2}\sqrt{\mathbb{E}\| \check{x}_k \|_{Q_1}^2}  + \right. \\&  \quad \left.\| \Delta_1 \|_{Q_1} \| Q_1 \|_{P_0} \sum_{s=0}^{k-1}\|A_1^{k-s-1}\| \sqrt{\mathbb{E}[ \| x_{s}\|_{P_0}^2 \mathbf{1}_{F_{s}} ]}\right)^2
    \end{align*}
    where by Cauchy-Schwarz inequality and Theorem~\ref{thm:moment_prob}, it holds $$\sqrt{\mathbb{E}\left[ \| x_{s}\|_{P_0}^2 \mathbf{1}_{F_{s}} \right]}  \leq (\mathbb{E}\| x_s \|_{P_0}^4)^{1/4}\mathbb{P}(F_s)^{1/4} \leq \mathcal{G}$$ for any $s$. Furthermore, we have $\mathbb{E}\| \check{x}_k \|_{Q_1}^2 \leq C_1^2$ guaranteed by~\eqref{eq:common_lyap}. Hence,
    \begin{equation*}
        \mathbb{E}\| x_k \|_{Q_1}^2  - \mathbb{E}\|\check{x}_k \|_{Q_1}^2 \leq 2 C_1 C_2 \mathcal{G} + C_2^2 \mathcal{G}^2.
    \end{equation*}
    
    \paragraph{Bounding $\mathbb{E}(x_k^T \Delta_2 x_k \mathbf{1}_{F_k})$} Notice $$x_k^T \Delta_2 x_k \mathbf{1}_{F_k} \leq \|\Delta_2\|\| x_k\|^2\mathbf{1}_{F_k} \leq C_3 \| x_k\|_{P_0}^2\mathbf{1}_{F_k}.$$ Following a similar argument to part (a), we can get $$\mathbb{E}(x_k^T \Delta_2 x_k \mathbf{1}_{F_k}) = C_3 \mathcal{G}^2.$$

    Combining the above two parts, we obtain the desired conclusion.
\end{proof}

The below corollary indicates that under proper choice of $t$, the performance loss caused by switching can decay super-exponentially $M$ is enlarged:

\begin{corollary}
    When $K_1$ is held constant, and $M, t$ are varied, it holds
    \begin{equation}
        J^{K_1,M,t} - J^{K_1} = O(t^{1/4}\exp(-cM^2))
        \label{eq:gap_bigO}
    \end{equation}
    as $M \rightarrow \infty, t \rightarrow \infty, t^{1/4} \exp \left(-c M^{2}\right) \rightarrow 0$, where $c = (1 - \rho^{1/4})^2 / (16 \| \tilde{W} \| \| P \| \| P^{-1} \|\| K_1 - K_0 \|^2)$ is a system-dependent constant.
    \label{cor:superexponential}
\end{corollary}

\section{Extension to Noise Distributions With Bounded Fourth-Order Moments}
\label{sec:extension}

In this section, instead of Gaussian distributed noise, the theoretical results are extended to the case where the process noise $\{w_k \}$ is i.i.d. according to a distribution that is heavier-tailed with bounded fourth-order moment.

\begin{assumption}
    The process noise $\{w_k \}$ is i.i.d. with:
    \begin{equation*}
        \mathbb{E} w_k = 0, \quad \mathbb{E}(w_k w_k^T) = W, \quad \mathbb{E} \| w_k \|^4 = \mu_4.
    \end{equation*}
    \label{assumption:noise}
\end{assumption}

On one hand, Theorem~\ref{thm:bounded_cost} still holds since its proof only requires $\mathbb{E} w_k = 0$ and $\mathbb{E}(w_k w_k^T) = W$. On the other hand, Theorems~\ref{thm:moment_prob} and~\ref{thm:gap}, which rely on the sub-Gaussian tail, need to be adjusted. The following theorem parallels Theorem~\ref{thm:moment_prob}:

\begin{theorem}
    Assume that $P_0, \rho_0$ satisfy~\eqref{eq:P0} and~\eqref{eq:rho0}, that $\rho, P, t$ satisfy~\eqref{eq:common_lyap}, and that Assumption~\ref{assumption:noise} holds. Let
    $$\tilde{W} = \sum_{\tau=0}^\infty (A+BK_0)^\tau W((A+BK_0)^\tau)^T$$
    and
    $$\tilde{\mu}_4 = \frac{\| P_0 \|^2 \mu_4}{1 - \rho_0^2} + \frac{2\rho_0 \tr({W} P_0)}{(1 - \rho_0^2)(1 - \rho_0)}.$$
    If $\| K_1 - K_0 \| \leq \mathcal{K}$, then the following statements hold:
    \begin{enumerate}
        \item The fourth moments of $x_k$ is bounded:
        \begin{equation}
            \mathbb{E} \| x_k \|^4_{P_0} \leq 8(\tilde{\mathcal{Q}} \| P_0 \|_P^2+ \tilde{\mu}_4),
        \end{equation}
        where 
        \begin{equation*}
            \tilde{\mathcal{Q}} = \frac{6\rho (\tr(\tilde{W}P)^2) + (1 - \rho) \| P \|_{P_0}^2 \tilde{\mu}_4}{(1 - \rho)(1 - \rho^2)},
        \end{equation*}
        \item The probability of not using feedback gain $K_1$ satisfies:
        \begin{equation*}
            \mathbb{P}(u_k \neq K_1 x_k) \leq t \mathcal{P}(M / \mathcal{K}),
        \end{equation*}
        where
        \begin{equation*}
            \mathcal{P}(a) = \frac{\|P \|_{P_0}^2 \tilde{\mu}_4}{(1 - \rho^{1/4})^4(1 - \rho)a^4},
        \end{equation*}
        which decays polynomially w.r.t. the threshold $M$.
    \end{enumerate}
    \label{thm:moment_prob_heavytail}
\end{theorem}

The following theorem parallels Theorem~\ref{thm:gap}:

\begin{theorem}
    Under Assumption~\ref{assumption:noise}, with $\rho_0, P_0, \rho, P, \tilde{W}, \mathcal{K}, \tilde{\mu}_4$, $\tilde{\mathcal{Q}}, \mathcal{P}$ defined the same as in Theorem~\ref{thm:moment_prob_heavytail}, and $C_1, C_2, C_3$ defined the same as in Theorem~\ref{thm:gap}, assuming that the dwell time $t$ satisfies \eqref{eq:common_lyap}, it holds that
    \begin{equation*}
        J^{K_1,M,t} - J^{K_1} \leq 2 C_1C_2\tilde{\mathcal{G}}+(C_2^2+C_3) \tilde{\mathcal{G}}^2,
    \end{equation*}
    where
    \begin{equation*}
        \tilde{\mathcal{G}} = 2^{3/4}(\tilde{\mathcal{Q}} \| P_0 \|_P^2+\tilde{\mu}_4)^{1/4}(t\mathcal{P}(M / \mathcal{K}))^{1/4}.
    \end{equation*}
    \label{thm:gap_heavytail}
\end{theorem}

\begin{proof}
    The proof parallels that of Theorem~\ref{thm:gap}, except that the bound on $(\mathbb{E}\| x_s \|_{P_0}^4)^{1/4}\mathbb{P}(F_s)^{1/4}$ should be $\tilde{\mathcal{G}}$ in place of $\mathcal{G}$.
\end{proof}

\begin{corollary}
    Under Assumption~\ref{assumption:noise}, when $K_1$ is held constant, and $M, t$ are varied, it holds
    \begin{equation}
        J^{K_1,M,t} - J^{K_1} = O(t^{1/4}/M)
        \label{eq:gap_bigO_heavytail}
    \end{equation}
    as $M \rightarrow \infty, t \rightarrow \infty, t^{1/4} / M \rightarrow 0$.
\end{corollary}

\section{Numerical Simulation}\label{sec:simulation}

This section demonstrates the safety guarantee and near-optimality of the proposed switching scheme by simulation on the Tennessee Eastman Process (TEP)~\cite{downs1993plant}, which is a commonly used process control system. In this simulation, we consider a simplified version of TEP similar to the one in~\cite{liu2020online} with full-state-feedback. The system has state dimension $n = 8$ and input dimension $m = 4$. The system is open-loop stable, and therefore the fallback controller is chosen as $K_0 = 0$. The LQ weight matrices are chosen as $Q=I,R=I$, and the process noise distribution is chosen to be $w_k \sim \mathcal{N}(0, I)$.

\subsection{Destabilizing $K_1$}

In this subsection, the primary feedback gain is chosen as $K_1 = K^* + 0.33 \mathbf{1}_m\mathbf{1}_n^T$, where $K^*$ is the optimal gain, such that $\rho(A+BK_1) \approx 1.01$. The trajectories of state norms with and without switching are compared in Fig.~\ref{fig:unstable_traj_compare}, from which it can be observed that the proposed switching strategy prevents the state from exploding exponentially.

\begin{figure}[!htbp]
    \begin{subfigure}{\columnwidth}
        \centering
        \includegraphics{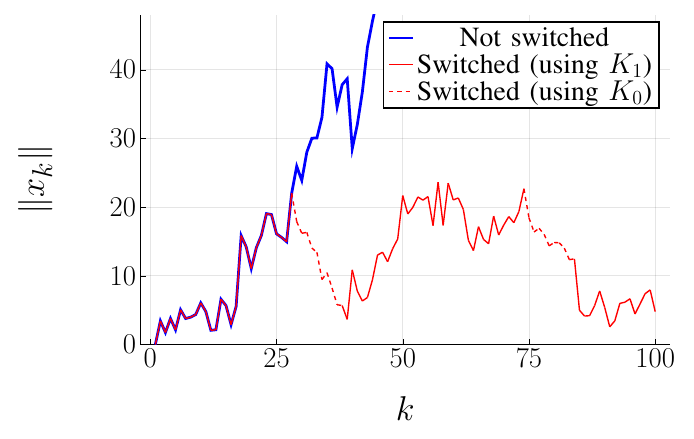}
        \caption{Destabilizing $K_1$}
        \label{fig:unstable_traj_compare}
    \end{subfigure}
    \begin{subfigure}{\columnwidth}
        \centering
        \includegraphics{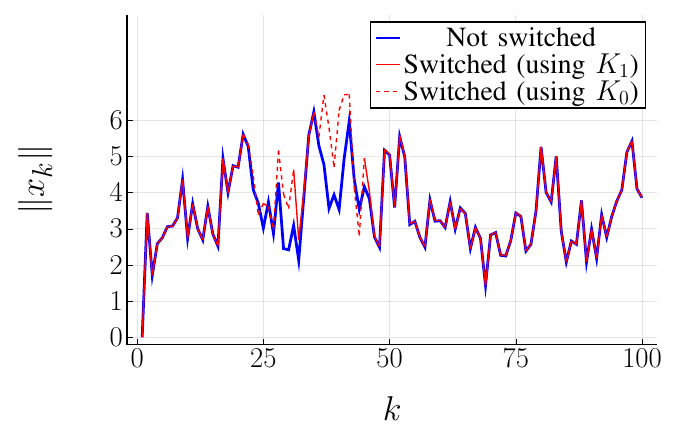}
        \caption{Stabilizing $K_1$}
        \label{fig:stable_traj_compare}
    \end{subfigure}
    \caption{Comparison of trajectories of state norms with and without switching, under the same realization of process noise. Parameters of the switching strategy are set to be $M = 1, t = 10$.}
\end{figure}

\subsection{Stabilizing $K_1$}

In this subsection, the primary feedback gain is chosen to be the optimal gain, i.e., $K_1 = K^*$. The trajectories of state norms with and without switching are compared in Fig.~\ref{fig:stable_traj_compare}. To quantify the relationship between the performance loss and the threshold $M$, we fix $K_1 = K^*$ and $t = 10$, and increase $M$ from $0.4$ to $3.1$. We evaluate the performance loss $J^{K^*,M,t} - J^*$ for each $M$, where $J^*$ is the optimal cost, by the empirical average of $10^5$ trajectories, each of which has a length of $10^3$. The empirical relative performance gap $(\hat{J}^{K^*,M,t} - J^*) / J^*$ against $M$ is plotted in a double-log plot in Fig.~\ref{fig:gap_wrt_M}. It can be observed that the performance gap converges to zero faster than a straight line (i.e., exponential convergence) as the switching threshold $M$ increases, which validates the super-exponential convergence property proved in Corollary~\ref{cor:superexponential}.

\begin{figure}[!htbp]
    \centering
    \includegraphics{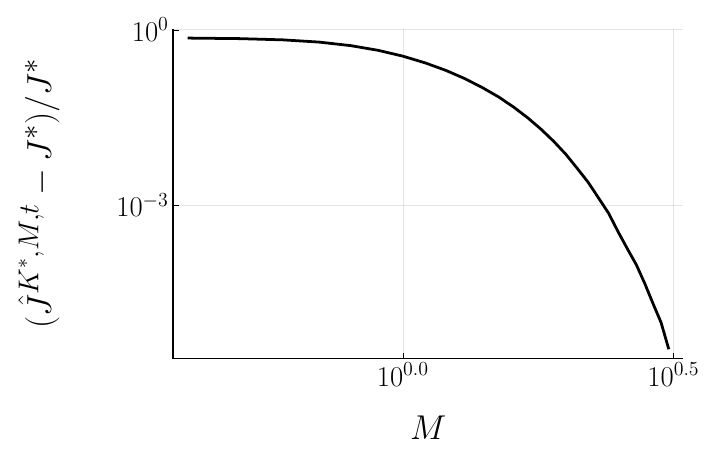}
    \caption{Double-log plot of relative performance gap against switching threshold $M$: super-exponential convergence to zero.}
    \label{fig:gap_wrt_M}
\end{figure}

\section{CONCLUSION}\label{sec:conclusion}

This paper introduces a plug-and-play switching strategy which enhances the safety of uncertified linear state-feedback controllers.
The strategy guarantees an upper bound on the LQ cost. Furthermore, the extra cost caused by switching as the switching threshold increases is quantified as decaying super-exponentially when the process noise is Gaussian, and decaying polynomially when the process noise obeys a heavy-tailed distribution with bounded fourth-order moments.
Future directions include extending the switching strategy with near-optimality guarantee to more general classes of systems.

\bibliographystyle{IEEEtran}
\bibliography{ref.bib}

\appendices
\section{Proof of Lemma~\ref{lemma:EV}}

\begin{proof}
    From the switching strategy, it holds
    \begin{equation*}
        x_{k+1} = A_0 x_k + d_k + w_k,
    \end{equation*}
    where $A_0 = A + BK_0$, and $d_k = B(K_1-K_0) x_k \mathbf{1}_{\{u_k = K_1 x_k\}}$ satisfies $\| d_k \| \leq M \| B \|$, and hence $\| d_k \|_{P_0} \leq M \| B \| \| P_0 \|^{1/2}$. Therefore, it holds
    \begin{align}
        & V_{0,k+1} = \| x_{k+1} \|_{P_0}^2 \leq (\| A_0 x_k \|_{P_0} + \| d_k + w_k \|_{P_0})^2  \nonumber \\
        &= (1+\sigma) \| A_0 x_k \|_{P_0}^2 + (1+\sigma^{-1}) \| d_k + w_k \|_{P_0}^2, \nonumber \\
        & \leq (1+\sigma)\rho_0 V_{0,k} + (1+\sigma^{-1}) \| d_k + w_k \|_{P_0}^2
        \label{eq:Vk1}
    \end{align}
    where $\sigma = (\rho_0^{-1} - 1) / 2$, and the last inequality follows from~\eqref{eq:P0}. Notice
    \begin{align*}
        & \mathbb{E} \| d_k + w_k \|_{P_0}^2 \leq \mathbb{E}(\| d_k \|_{P_0} + \| w_k \|_{P_0})^2 \\
        \leq & 2(\mathbb{E} \| d_k \|_{P_0}^2 +\mathbb{E}\| w_k \|_{P_0}^2) \leq 2(M \| B \|^2 \|P_0 \| + \tr(WP_0))=: \mathcal{B},
    \end{align*}
    where the last inequality follows from the fact that $\mathbb{E}\| w_k \|_{P_0}^2 = \tr(\mathbb{E}(w_kw_k^T)P_0) = \tr(WP_0)$. Therefore, it follows from~\eqref{eq:Vk1} that
    \begin{equation*}
        \mathbb{E} V_{0,k+1} \leq (1+\sigma)\rho_0 \mathbb{E}V_{0,k} + (1 + \sigma^{-1})\mathcal{B} ,
    \end{equation*}
    and therefore by induction on $k$ it holds
    \begin{equation}
        \mathbb{E} V_{0,k} \leq \frac{(1 + \sigma^{-1})\mathcal{B}}{1 - (1+\sigma)\rho_0}.
        \label{eq:V0_bound}
    \end{equation}
    Substituting the expressions for $\sigma$ and $\mathcal{B}$ into~\eqref{eq:V0_bound} leads to the conclusion.
\end{proof}
\section{Proof of Theorem~\ref{thm:exp_weight_sum}}

\begin{proof}
    Let us choose $\sigma = \varrho^{1/2}$, then for any $a > 0$, considering the fact that $a = \sum_{i=0}^\infty \sigma^{i}(1-\sigma)a$, it holds
    \begin{align*}
	&\{ S_k \geq a \}  \subseteq \left\{ \sum_{i = 0}^k \varrho^{k - i} X_i \geq \sum_{i = 0}^k \sigma^{k - i} (1 - \sigma) a \right\} \\
	& \subseteq \bigcup_{i = 0}^k \{ \varrho^{k - i} X_i \geq \sigma^{k - i} (1 - \sigma) a \}  = \bigcup_{i = 0}^k \{  X_i \geq (\sigma / \varrho)^{k - i} (1 - \sigma) a \}
    \end{align*}
    and hence,
    $$
    \mathbb{P}(\| S_k \| \geq a) \leq C_1 \sum_{i = 0}^\infty \beta^{\alpha^i},
    $$
    where $\beta = \exp(-(1 - \sigma)^2 C_2 a^2 ) \in (0,1/2]$ since $a\geq 2C_2^{-1/2}(1-\sigma)^{-1}$, and $\alpha = (\sigma/\varrho)^2 = \rho^{-1} > 1$.  Next we only need to prove
    $$S:=\sum_{i = 0}^\infty \beta^{\alpha^i} \leq \frac{2\beta}{\min\{\alpha - 1, 1\}}.$$
    By Bernoulli's inequality, for $i \geq 1$,
    $$
    \alpha^i - 1 = (1 + \alpha - 1)^i - 1 \leq 1 + i(\alpha - 1) - 1 = i(\alpha - 1),
    $$
    and $\alpha^i - 1 \leq i(\alpha - 1)$ also holds for $i = 0$. Hence
    $$
    S = \beta \sum_{i=0}^\infty \beta^{\alpha^i - 1} \leq \beta\sum_{i=0}^\infty \beta^{i(\alpha - 1)} = \frac{\beta}{1 - \beta^{\alpha - 1}}.
    $$
    When $0 < \beta \leq 1/2$ and $1 < \alpha < 2$, by Bernoulli's inequality,
    $$
    \beta^{\alpha - 1} = (1 + \beta - 1)^{\alpha - 1} \leq 1 + (\alpha - 1)(\beta - 1),
    $$
    and hence,
    $$
    1 - \beta^{\alpha - 1} \geq (\alpha - 1)(1 - \beta) \leq (\alpha - 1) / 2,
    $$
    which implies $S \leq 2\beta / (\alpha - 1)$.
    Noticing that $S$ decreases monotonically as $\alpha$ increases, when $\alpha > 1$, we always have $S \leq 2\beta / \min\{ \alpha - 1, 1\}$.
\end{proof}

\section{Proof of Theorem~\ref{thm:moment_prob}}
\label{app:opt}

\subsection{Supporting lemmas}

By combining the $t$ consecutive steps of applying the fallback gain into a single step, we can transform the original system into a new linear time-varying system which is stable with a common Lyapunov function defined by $\rho, P$ in~\eqref{eq:common_lyap}. To be specific, denote the state sequence of the transformed system by $\{ \tilde{x}_j = x_{i(j)} \}$, which is a subsequence of the state sequence $\{ x_k \}$ of the original closed-loop system, indexed by
\begin{align}
    i(0) = 0, i(j + 1) = \begin{cases}
        i(j) + 1 & u_{i(j)} = K_1 x_{i(j)}, \\
        i(j) + t & \text{otherwise}.
    \end{cases}
    \label{eq:subseq_notation}
\end{align}
It follows that the transformed system evolves as
\begin{equation}
    \tilde{x}_{j+1} = \tilde{A}_j \tilde{x}_j + \tilde{w}_j,
    \label{eq:transformed_sys}
\end{equation}
where $\tilde{A}_j, \tilde{w}_j$ are defined as:
\begin{align}
    & \tilde{A}_j = \begin{cases}
        A+BK_1 & u_{i(j)} = K_1 x_{i(j)}, \\
        (A+BK_0)^t, & \text{otherwise}, \\
    \end{cases} \label{eq:Atk}\\
    & \tilde{w}_j = \begin{cases}
        w_{i(j)} & u_{i(j)} = K_1 x_{i(j)}, \\
        \sum_{\tau = 1}^t (A+BK_0)^{t - \tau} w_{i(j) + \tau - 1} &  \text{otherwise}.
    \end{cases} \label{eq:wtk}
\end{align}

The next two lemmas are two properties of the transformed system that will pave the way for proving Theorem~\ref{thm:moment_prob}:

\begin{lemma}
    Let $\tilde{V}_j = \tilde{x}_j^T P \tilde{x}_j$, then it holds for any $j$ that
    \begin{equation}
        \mathbb{E} \tilde{V}_{j}^{2}\leq \mathcal{Q},
        \label{eq:EV2}
    \end{equation}
    where $\mathcal{Q}$ is defined in Theorem~\ref{thm:moment_prob}.
    \label{lemma:EV2}
\end{lemma}

\begin{proof}
    From~\eqref{eq:common_lyap} and~\eqref{eq:Atk}, it follows that
    \begin{equation}
        \tilde{V}_{j+1} \leq \rho \tilde{V}_j + \eta_j,
        \label{eq:recursive_V}
    \end{equation}
    where $\eta_j = 2\tilde{w}_j^T P \tilde{A}_j \tilde{x}_j + \tilde{w}_j^T P \tilde{w}_j$. From $\mathbb{E}(\tilde{w}_j^T P \tilde{A}_j \tilde{x}_j) = \mathbb{E}(\mathbb{E}(\tilde{w}_j | \tilde{x}_j)^T P \tilde{A}_j \tilde{x}_j) = 0$, it follows that $\mathbb{E}\eta_j = \tr(\mathbb{E}(\tilde{w}_j\tilde{w}_j^T)P) \leq \tr(\tilde{W}P)$, and hence
    \begin{equation}
        \mathbb{E}\tilde{V}_j \leq \tr(\tilde{W}P) / (1 - \rho).
        \label{eq:EVk}
    \end{equation}
    To proceed, we square and take the expectations on both sides of~\eqref{eq:recursive_V}, and obtain
    \begin{equation}
        \mathbb{E} \tilde{V}_{j+1}^2 \leq \rho^2 \mathbb{E} \tilde{V}_j^2 + 2\rho \mathbb{E}(\tilde{V}_j\eta_j) + \mathbb{E}\eta_j^2.
        \label{eq:recursive_V2}
    \end{equation}
    \paragraph{Bound on $\mathbb{E}(\tilde{V}_j\eta_j)$}
    \begin{equation*}
        \mathbb{E}(\tilde{V}_j\eta_j) = 2\mathbb{E}(\tilde{w}_j^TP\tilde{A}_j\tilde{x}_j\tilde{V_j}) + \mathbb{E}(\tilde{w}_j^T P \tilde{w}_j\tilde{V}_j),
    \end{equation*}
    where:
    \begin{itemize}
        \item $\mathbb{E}(\tilde{w}_j^T P \tilde{A}_j \tilde{x}_j\tilde{V}_j) = \mathbb{E}(\mathbb{E}(\tilde{w}_j | \tilde{x}_j)^T P \tilde{A}_j \tilde{x}_j\tilde{V}_j) = 0$;
        \item $\mathbb{E}(\tilde{w}_j^T P \tilde{w}_j\tilde{V}_j) = \tr(\mathbb{E}(\tilde{V}_j)\mathbb{E}(\tilde{w}_j\tilde{w}_j^T)P) = \tr(\tilde{W}P)\cdot\mathbb{E}\tilde{V}_j$, since $\tilde{w}_j$ is independent of $\tilde{V}_j$.
    \end{itemize}
    Hence,
    \begin{equation}
        \mathbb{E}(\tilde{V}_j \eta_j) \leq \tr(\tilde{W}P)\mathbb{E} \tilde{V}_j .
        \label{eq:EVeta}
    \end{equation}
    \paragraph{Bound on $\mathbb{E} \eta_j^2$}
    \begin{align*}
        \mathbb{E}\eta_j^2 = & 4 \mathbb{E}(\tilde{x}_j^T \tilde{A}_j^T P \tilde{w}_j \tilde{w}_j^T P \tilde{A}_j \tilde{x}_j) + \\
        & 4\mathbb{E}(\tilde{w}_j^T P \tilde{A}_j \tilde{x}_j \tilde{w}_j^T P \tilde{w}_j) + \mathbb{E}(\tilde{w}_j^T P \tilde{w}_j \tilde{w}_j^T P \tilde{w}_j),
    \end{align*}
    where we can bound each term respectively as follows: \begin{itemize}
        \item $\mathbb{E}(\tilde{x}_j^T \tilde{A}_j^T P \tilde{w}_j \tilde{w}_j^T P \tilde{A}_j \tilde{x}_j) = \tr(\mathbb{E}(\tilde{w}_j \tilde{w}_j^T)P)\cdot \mathbb{E}(\tilde{x}_j^T \tilde{A}_j^T P \tilde{A}_j \allowbreak \tilde{x}_j) \leq \rho\tr(\tilde{W}P) \mathbb{E} \tilde{V}_j$ since $\tilde{w}_j$ is independent of $\tilde{x}_j$ and $\tilde{A}_j$;
        \item $\mathbb{E}(\tilde{w}_j^T P \tilde{A}_j \tilde{x}_j \tilde{w}_j^T P \tilde{w}_j) = \tr\{\mathbb{E}[\tilde{A}_j \tilde{x}_j \mathbb{E}( \tilde{w}_j^T P \tilde{w}_j \tilde{w}_j^T\mid x_j) ]\} = 0$ by symmetry;
        \item $\mathbb{E}(\tilde{w}_j^T P \tilde{w}_j \tilde{w}_j^T P \tilde{w}_j) = \mathbb{E}\| \tilde{w}_j \|_{P}^4 \leq \|P \|_{\tilde{W}^{-1}}^2 \cdot \mathbb{E}  \| \tilde{w}_j \|_{\tilde{W}^{-1}}^4 \leq \|P \|_{\tilde{W}^{-1}}^2 \mathbb{E} \nu^2 = (n^2 + 2n)\|P \|_{\tilde{W}^{-1}}^2$, where $\nu \sim \chi^2(n)$.
    \end{itemize}
    Hence,
    \begin{equation}
        \mathbb{E} \eta_j^2 \leq 4\rho\tr(\tilde{W}P)\mathbb{E}\tilde{V}_j + (n^2 + 2n) \| P \|_{\tilde{W}^{-1}}^2.
        \label{eq:Eeta2}
    \end{equation}

    The conclusion follows from substituting~\eqref{eq:EVeta},~\eqref{eq:Eeta2} and~\eqref{eq:EVk} into~\eqref{eq:EV2} and applying induction.
\end{proof}

\begin{lemma}
    For $a \geq a_0$, it holds for any $j$ that
    \begin{equation*}
        \mathbb{P}\left(\left\|\tilde{x}_{j}\right\| \geq a\right) \leq  \mathcal{E}(a),
    \end{equation*}
    where $a_0, \mathcal{E}(a)$ are defined in Theorem~\ref{thm:moment_prob}.
    \label{lemma:escape_prob}
\end{lemma}

\begin{proof}
    Notice $\tilde{x}_j = \sum_{s=0}^{j-1}(\prod_{r=s+1}^{k-1} \tilde{A}_r) \tilde{w}_s$. From~\eqref{eq:common_lyap} and~\eqref{eq:Atk}, it follows that
    \begin{align*}
        \| \tilde{x}_j \|_P &= \left\| \sum_{s=0}^{j-1} P^{1/2}\left(\prod_{r=s+1}^{j-1} \tilde{A}_r\right) \tilde{w}_s \right\| \\
        & \leq \left\| \sum_{s=0}^{j-1} \rho^{(j-s-1)/2} P^{1/2}\tilde{w}_s \right\| \leq \sum_{s=0}^{j-1} \rho^{(j-s-1)/2} \| \tilde{w}_s \|_P.
    \end{align*}
    By~\eqref{eq:wtk}, it holds $\tilde{w}_s | \mathcal{F}_{s - 1}\sim \mathcal{N}(0, W_s)$, where $\mathcal{F}_{s - 1}$ is the $\sigma$-algebra generated by $\tilde{w}_0, \ldots, \tilde{w}_{s - 1}$, and $W_s \in \{ W, 
    \sum_{\tau=0}^{t-1}( A+B K_{0})^{\tau} W\left(\left(A+B K_{0}\right)^{\tau}\right)^{T}
     \}$; in either case, it holds $W_s \preceq \tilde{W}$. Hence, by a concentration bound on Gaussian random vectors~\cite[Lemma 3.1]{ledoux1991probability}, it holds for any $s$ and any $a > 0$ that
     $$
     \mathbb{P}(\| \tilde{w}_s \|_P \geq a) \leq 2n\exp(-a^2 / (2n \| \tilde{W} \| \| P \|)).
     $$
    Invoking Theorem~\ref{thm:exp_weight_sum} with $\varrho = \rho^{1/2}$, and assuming w.l.o.g. that $\rho \in (1/4, 1)$, it follows that
    \begin{equation*}
        \mathbb{P}(\| \tilde{x}_j \|_P \geq a) \leq \frac{4n}{\rho^{-1 / 2}-1} \exp \left(-\frac{(1-\rho^{1 / 4})^{2}}{2n \| \tilde{W} \| \| P \|} a^{2}\right)
    \end{equation*}
    for any $a\geq a_0$. 

    Meanwhile, it holds
    \begin{equation*}
        \{ \| \tilde{x}_j \| \geq a \} \subseteq \{ \| \tilde{x}_j \|_{P} \geq a \| P^{-1} \|^{-1/2} \},
    \end{equation*}
    from which the conclusion follows.
\end{proof}

\subsection{Proof of Theorem~\ref{thm:moment_prob}}
\label{sec:proof_moment_prob}

Now we are ready to prove Theorem~\ref{thm:moment_prob}, whose contents are restated below:
\begin{align}
    & \mathbb{E} \| x_k \|^4_{P_0} \leq 8(\mathcal{Q} \| P_0 \|_P^2+(n^2 + 2n)\| P_0 \|_{\tilde{W}^{-1}}^2), \label{eq:Ex4}\\
    & \mathbb{P}(u_k \neq K_1 x_k) \leq t \mathcal{E}(M / \mathcal{K}). \label{eq:Pu0}
\end{align}

\begin{proof}
    The proof is devoted to translating properties of the transformed state sequence $\{ \tilde{x}_j \}$ back into properties of the original state sequence $\{ x_k \}$. In what follows we shall prove~\eqref{eq:Ex4} and~\eqref{eq:Pu0} respectively.

    \paragraph{Proof of~\eqref{eq:Ex4}} Let $j = \sup\{ s \in \mathbb{N} \mid i(s) \leq k \}$, i.e., $\tilde{x}_j$ is the last state in the transformed state sequence that occurs no later than $x_k$. Consequently,
    \begin{equation*}
        x_k = (A+BK_0)^{k - i(j)} \tilde{x}_j + \tilde{w}_j.
    \end{equation*}
    From~\eqref{eq:P0}, it follows that
    \begin{equation*}
        \| x_k \|_{P_0} \leq \rho^{(k - i(j))/2}\| \tilde{x}_j \|_{P_0} + \| \tilde{w}_j \|_{P_0} \leq \| \tilde{x}_j \|_{P_0} + \| \tilde{w}_j \|_{P_0}.
    \end{equation*}
    Hence, applying the power means inequality $((a + b) / 2)^4 \leq (a^4 + b^4) / 2$, and taking the expectation on both sides, we have
    \begin{equation*}
        \mathbb{E}\| x_k \|_{P_0}^4 \leq 8(\mathbb{E}\| \tilde{x}_j \|_{P_0}^4 + \mathbb{E}\| \tilde{w}_j \|_{P_0}^4),
        \label{eq:xP04}
    \end{equation*}
    where:
    \begin{itemize}
        \item $\mathbb{E} \| \tilde{x}_j \|_P^4 \leq \mathcal{Q}$ by Lemma~\ref{lemma:EV2}, and hence $\mathbb{E} \| \tilde{x}_j \|_{P_0}^4 \leq \mathcal{Q} \| P_0 \|_P^2$;
        \item $\mathbb{E}\| \tilde{w}_j \|_{P_0}^4 \leq \|P_0\|_{\tilde{W}^{-1}}^2 \mathbb{E}\| \tilde{w}_j \|_{\tilde{W}^{-1}}^4 \leq \|P_0\|_{\tilde{W}^{-1}}^2 \mathbb{E}\nu^2 = (n^2 + 2n)\|P_0\|_{\tilde{W}^{-1}}^2$, where $\nu \sim \chi^2(n)$.
    \end{itemize}
    Combining the above two items leads to the conclusion.

    \paragraph{Proof of~\eqref{eq:Pu0}} Let $I = \{ k \in \mathbb{N} | \exists j\in \mathbb{N} \text{ s.t. } i(j) = k\}$, i.e., $I$ is the index set for states that occur in the transformed state sequence.
    A sufficient and necessary condition for $u_k \neq K_1 x_k$ is that exactly one of $x_k,x_{k-1}, \ldots, x_{k-t+1}$ belongs to the transformed state sequence and triggers the switching rule, and hence,
    \begin{equation*}
        \{ u_k \neq K_1 x_k \} \subseteq \bigcup_{\tau = 0}^{t - 1} \{ \| (K_1 - K_0) x_{k - \tau} \|\geq M, k - \tau \in I \}.
    \end{equation*}
    For each event in the RHS above, we have
    \begin{align}
        & \mathbb{P}( \| (K_1 - K_0) x_{k - \tau} \| \geq M, k - \tau \in I)  \nonumber \\
        = & \mathbb{P}( \| x_{k - \tau} \| \geq M / \mathcal{K} \mid k - \tau \in I)\,\mathbb{P}(k - \tau \in I) \nonumber \\
        \leq & \mathbb{P}( \| x_{k - \tau} \| \geq M / \mathcal{K} \mid k - \tau \in I). \nonumber
    \end{align}
    Since $\mathbb{P}(\| \tilde{x}_j \| \geq M / \mathcal{K}) \leq \mathcal{E}(M / \mathcal{K})$ for any $j$ according to Lemma~\ref{lemma:escape_prob}, and $k - \tau \in I$ indicates $x_{k - \tau}$ belongs to $\{\tilde{x}_j \}$, it follows that $\mathbb{P}( \| x_{k - \tau} \| \geq M / \mathcal{K} \mid k - \tau \in I) \leq \mathcal{E}(M / \mathcal{K})$.
    Taking the union bound over $\tau = 0,1,\ldots,t-1$, we reach the conclusion.
\end{proof}

\section{Proof of Theorem~\ref{thm:moment_prob_heavytail}}

In this appendix, we adopt the same definition of $\{\tilde{x}_j\},\{\tilde{A}_j\}$ and $\{\tilde{w}_j\}$ as in~\eqref{eq:subseq_notation} to~\eqref{eq:wtk}.

\subsection{Supporting lemmas}

\begin{lemma}
    Under Assumption~\ref{assumption:noise}, it holds
    \begin{equation*}
        \mathbb{E}\| \tilde{w}_j \|_{P_0}^4 \leq \tilde{\mu}_4 := \frac{\| P_0 \|^2 \mu_4}{1 - \rho_0^2} + \frac{2\rho_0 \tr({W} P_0)}{(1 - \rho_0^2)(1 - \rho_0)}.
    \end{equation*}
    \label{lemma:fourth_order_moment_of_response}
\end{lemma}

\begin{proof}
    Let $v = \sum_{i=0}^{t-1} A^i v_i$, where $v_i \stackrel{d}{=}w_1$ independently.
    According to the definition of $\{\tilde{w}_j\}$ in~\eqref{eq:wtk}, it holds $\mathbb{E}\| \tilde{w}_j \|_{P_0}^4 \leq \mathbb{E}\| v \|_{P_0}^4$ for any $j$. It holds for the above defined $v$ that
    \begin{align*}
        \| v \|_{P_0}^4 &= (v^T P_0 v)^2 \\
        &= \sum_{i=0}^{t-1}\sum_{j=0}^{t-1}\sum_{k=0}^{t-1}\sum_{l=0}^{t-1}(v_i^T(A^i)^T P_0 A^j v_j)(v_k^T(A^k)^T P_0 A^l v_l), \\
        &= \sum_{i=0}^{t-1} (v_i^T(A^i)^T P_0 A^i v_i)^2 + \\ &\quad \sum_{i=0}^{t-1}\sum_{\substack{j=0 \\ j\neq i}}^{t-1}(v_i^T(A^i)^T P_0 A^i v_i)(v_j^T(A^j)^T P_0 A^j v_j) + L \\
        &\leq \sum_{i=0}^{t-1} \rho_0^{2i}\|v_i \|_{P_0}^4 + \sum_{i=0}^{t-1}\sum_{\substack{j=0 \\j\neq i}}^{t-1} \rho_0^{i+j}  \| v_i \|_{P_0}^2 \| v_j \|_{P_0}^2 + L,
    \end{align*}
    where $L$ consists of terms that are linear w.r.t. at least one of $v_i$, and hence $\mathbb{E}L = 0$ since $\mathbb{E} v_i = 0$ and $\{v_i \}$ are mutually independent.
    Therefore,
    \begin{align*}
        & \mathbb{E}\| v \|_{P_0}^4 \leq
        \sum_{i=0}^{t-1} \rho_0^{2i} \mathbb{E}\|v_i \|_{P_0}^4 + \sum_{i=0}^{t-1}\sum_{\substack{j=0 \\ j\neq i}}^{t-1}\rho_0^{i+j} \mathbb{E} \| v_i \|_{P_0}^2 \mathbb{E}\| v_j \|_{P_0}^2\\
        &\leq \sum_{i=0}^\infty \rho_0^{2i} \|P_0\|^2 \mu_4 + \sum_{i=0}^\infty 2i(\rho_0^{2i-1}+\rho_0^{2i})\tr(WP_0) = \tilde{\mu}_4,
    \end{align*}
    and hence $\mathbb{E} \| \tilde{w}_j \|_{P_0}^4 \leq \tilde{\mu}_4$ for any $j$.
\end{proof}

\begin{lemma}
    Under Assumption~\ref{assumption:noise}, let $\tilde{V}_j = \tilde{x}_j^T P \tilde{x}_j$, then it holds for any $j$ that
    \begin{equation}
        \mathbb{E} \tilde{V}_{j}^{2}\leq \tilde{\mathcal{Q}},
        \label{eq:EV2_heavytail}
    \end{equation}
    where $\tilde{\mathcal{Q}}$ is defined in Theorem~\ref{thm:moment_prob_heavytail}.
    \label{lemma:EV2_heavytail}
\end{lemma}

\begin{proof}
    This proof parallels that of Lemma~\ref{lemma:EV2}, and the only difference is the bound on $\mathbb{E}(\tilde{w}_j^T P \tilde{w}_j \tilde{w}_j^T P \tilde{w}_j)$: now by Lemma~\ref{lemma:fourth_order_moment_of_response}, it holds
    \begin{equation*}
        \mathbb{E}(\tilde{w}_j^T P \tilde{w}_j \tilde{w}_j^T P \tilde{w}_j) = \mathbb{E} \| \tilde{w}_j \|_P^4\leq \|P\|_{P_0}^2 \tilde{\mu}_4,
    \end{equation*}
    from which the conclusion follows.
\end{proof}

\begin{lemma}
    Under Assumption~\ref{assumption:noise}, for any $a > 0$, it holds for any $j$ that
    \begin{equation*}
        \mathbb{P}\left(\left\|\tilde{x}_{j}\right\| \geq a\right) \leq  \mathcal{P}(a),
    \end{equation*}
    where $\mathcal{P}(a)$ is defined in Theorem~\ref{thm:moment_prob_heavytail}.
    \label{lemma:escape_prob_heavytail}
\end{lemma}

\begin{proof}
    Similarly to the proof of Lemma~\ref{lemma:escape_prob}, it holds
    \begin{equation*}
        \| \tilde{x}_j \|_P \leq \sum_{s = 0}^{j-1} \rho^{(j-s-1)/2} \| \tilde{w}_s \|_P.
    \end{equation*}
    By Markov's inequality, for any $a > 0$, it holds
    \begin{equation*}
        \mathbb{P}(\| \tilde{w}_s \|_P \geq a) \leq \mathbb{E}\| \tilde{w}_s \|^4 / a^4 \leq \tilde{\mu}_4 \| P \|_{P_0}^2 / a^4,
    \end{equation*}
    where the last inequality follows from Lemma~\ref{lemma:fourth_order_moment_of_response}.
    Now let $\sigma  = \rho^{1/4}$. Considering the fact that $a \geq \sum_{s=0}^{j-s-1}\sigma^{j-s-1}(1-\sigma)a$, it holds
    \begin{align*}
        &\mathbb{P}(\| \tilde{x}_j \|_P \leq a) \leq  
        \sum_{s=0}^{j-1} \mathbb{P}(\rho^{(j-s-1)/2} \| \tilde{w}_s \|_P \geq \sigma^{j-s-1} (1 - \sigma) a) \\
        \leq & \tilde{\mu}_4 \| P \|_{P_0}^2 (1 - \rho^{1/4})^{-4}a^{-4} \sum_{s=0}^{j - 1} \rho^{j-s-1} \\
        \leq & \tilde{\mu}_4 \| P \|_{P_0}^2 (1 - \rho^{1/4})^{-4}(1 - \rho)^{-1}a^{-4}.
    \end{align*}
\end{proof}

\subsection{Proof of Theorem~\ref{thm:moment_prob_heavytail}}

\begin{proof}
    This proof parallels Theorem~\ref{thm:moment_prob}, but the following bounds need to be updated:
    \begin{itemize}
        \item $\mathbb{E} \| \tilde{x}_j \|_{P_0}^4 \leq \tilde{Q}\| P_0 \|_P^2$ for any $j$, according to Lemma~\ref{lemma:EV2_heavytail};
        \item $\mathbb{E} \| \tilde{w}_j \|_{P_0}^4 \leq \tilde{\mu}_4$ for any $j$, according to Lemma~\ref{lemma:fourth_order_moment_of_response};
        \item $\mathbb{P}(\| x_{k - \tau} \| \geq M / \mathcal{K} \mid k - \tau \in I) \leq \mathcal{P}(M / \mathcal{K})$ for any $k$ and $\tau$, according to Lemma~\ref{lemma:escape_prob_heavytail}.
    \end{itemize}
    Substituting the above bounds into the proof of Theorem~\ref{thm:moment_prob} in Appendix~\ref{sec:proof_moment_prob} leads to the conclusion.
\end{proof}

\end{document}